\documentclass[reqno,10pt]{amsart}

\usepackage{amscd,amsmath,amsfonts,amssymb,amsthm,cite,mathrsfs,color}
\usepackage{dsfont} \usepackage{ leftidx}

\newtheorem{theorem}{Theorem}[section]

\newtheorem{proposition}[theorem]{Proposition}

\theoremstyle{definition}
\newtheorem{remark}[theorem]{Remark}
\theoremstyle{definition}
\newtheorem {definition}[theorem]{Definition}
\theoremstyle{remark}
\newtheorem*{acknow}{Acknowledgments}

\numberwithin{equation}{section}
\numberwithin{theorem}{section}

\def\be{\begin{equation}}
\def\ee{\end{equation}}
\def\bae{\begin{eqnarray}}
\def\eae{\end{eqnarray}}

\def\TR{\mathrm{Tr}}

\begin{document}

\title[Correlations for the Gaussian $\beta$-ensemble]{Scaling limits of correlations of characteristic polynomials for the Gaussian $\beta$-ensemble with external source}

\author{Patrick Desrosiers} \address{Instituto Matem\'atica y F\'isica,
Universidad de Talca, 2 Norte 685, Talca,
Chile}\email{patrick.desrosiers@inst-mat.utalca.cl}
\curraddr{CRIUSMQ, 2601 de la Canardi\`ere, Qu\'ebec, Canada, G1J 2G3}
\email{patrick.desrosiers.1@ulaval.ca}
\author{Dang-Zheng Liu} \address{School of Mathematical Sciences, University of Science and Technology of China, Hefei,
230026, P.R. China \&\
Wu Wen-Tsun Key Laboratory of Mathematics,
University of Science and Technology of China, Chinese Academy of Sciences,
Hefei, 230026
P.R. China
}
\email{dzliu@ustc.edu.cn}

\date{February  2014}
 \keywords{Random matrices, Beta-ensembles, External source, Phase transition, Jack polynomials}

  \subjclass[2010]{15B52}

\begin{abstract}

We study the averaged product of characteristic polynomials of large random matrices in the Gaussian $\beta$-ensemble perturbed by an external source of finite rank.   We prove that at the edge of the spectrum,  the limiting correlations involve two families of multivariate functions of Airy and Gaussian types.   The precise form of the limiting correlations depends on the strength of the nonzero eigenvalues of the external source.  A critical value for the latter is obtained and a phase transition phenomenon similar to  that of \cite{bbp} is established.     The derivation of our results relies mainly on  previous articles by the authors, which deal with duality formulas \cite{des} and asymptotics for Selberg-type integrals\cite{dl}.
\end{abstract}

\maketitle

\small
\tableofcontents
\normalsize

\newpage
\section{Introduction}

\subsection{Gaussian  ensembles with source}

Let $\mathbf{X}$ and $\mathbf{F}$  be $N\times N$  hermitian matrices
 with either real ($\beta=1$), complex ($\beta=2$) or quaternion real ($\beta=4$) entries.  We say that $\mathbf{X}$  belongs to the Gaussian ensemble with external source  $\mathbf{F}$  if it is randomly distributed according to a probability density function proportional to
\begin{equation}\label{density1}
\exp\left\{-\frac{\beta}{2}\TR(\mathbf{X}-\mathbf{F})^{2}\right\}.
\end{equation}
Note that the Gaussian ensemble with external source  is also called the shifted mean Gaussian ensemble. Obviously, when $\mathbf{F}$ is the null matrix, the three classical Gaussian ensembles -- that is  GOE ($\beta=1$), GUE ($\beta=2$), and  GSE ($\beta=4$)-- are recovered.

Now let  $x=(x_1,\ldots,x_N)$ and $f=(f_1,\ldots, f_N)$ respectively denote the eigenvalues of  $\mathbf{X}$ and $\mathbf{F}$.  Then, the use of standard techniques in random matrix theory  and Jack polynomial theory  (see \cite[Chapters 1 \& 13]{forrester} , \cite[Chapter VII]{macdonald} ) allows to show that the probability density for the eigenvalues of $\mathbf{X}$ is equal to
\begin{eqnarray}\label{eqbetadensities}
\frac{1}{G_{\beta,N}}\exp\Big\{ - \frac{\beta}{2}\sum_{i=1}^{N}( x^{2}_i+f^{2}_i) \Big\} \,\prod_{1\leq j<k\leq
N}|x_{j}-x_{k}|^{\beta}\,  {{\phantom{k}}_0\mathcal{F}_0}^{\!(2/\beta)}(\beta x;f).\label{PDFforgauss}
\end{eqnarray}
 The normalization constant is a special case  of Selberg's celebrated formula and is given in  Appendix \ref{appendix}, while $\!{\phantom{j}}_0\mathcal{F}_0^{ }$ is a multivariate hypergeometric function of exponential type whose exact expansion in terms of the Jack polynomials is known explicitly (see Section \ref{hypergeometric}).  When $f=(0,\ldots,0)$, the density \eqref{eqbetadensities} defines the now standard Gaussian $\beta$-ensemble of random matrices (see for instance, \cite{de,df2,bey,rrv,vv} and \cite[Section 1.9]{forrester}).

 We stress that for all $f\in\mathbb{R}^N$ and $\beta>0$, Eq.\ \eqref{eqbetadensities} provides a well-defined density, even if $\beta\neq 1,2,4$.  It is thus natural to define, like in \cite{des,for2012,wang},  the Gaussian $\beta$-ensemble with external source  as the set of real random variables  $x=(x_1,\ldots,x_N)$ distributed according to the density \eqref{eqbetadensities}.

A typical problem in random matrix theory is to determine the influence of  $f_1,\ldots, f_N$   on the distribution of  $x_1,\ldots,x_N$   as $N\to\infty$.   This was first addressed in physics in the case where the matrices are real ($\beta=1$) and all the entries of $\mathbf{F}$ are equal to $\mu$, so that $f=(N\mu,0,\ldots,0)$.  Indeed, in the mid 1960s, Lang \cite{lang} gave theoretical arguments in favor of a phenomenon first noticed by Porter with the help of numerical simulations: if $\mu$ is large enough,  then one eigenvalue of $\mathbf{X}$   separates from the main support for the eigenvalues, which is $[-\sqrt{2N}, \sqrt{2N}]$.  Jones et al. (see \cite{jones}) later proved the existence of critical value  $\mu_c=(2N)^{-1/2}$ that splits the statistical behavior of the eigenvalues of $\mathbf{X}$ into two separate phases:  if $\mu<\mu_c$, then the eigenvalues are distributed as if $\mu=0$, while if $\mu>\mu_c$,  then one eigenvalue completely separates from the others.

Few progresses were made on the distribution of the eigenvalue $x_1,\ldots,x_N$  in the real case ($\beta=1$) until the very recent works of Bloemendal, Vir\'ag, Mo, and Wang \cite{bv1,bv2,mo,wang}.  In fact, the latter references were  motivated, to a large extent, by the recent breakthroughs in the complex case ($\beta=2$) \cite{bbp}.

The ensembles with external source in the $\beta=2$ case  are   indeed very special, since  the
Harish-Chandra-Itzykson-Zuber integral formula provides the following compact expression:
\begin{equation}\label{iz}\frac{i^{N(N-1)/2}}{{1!2!\cdots (N-1)!}}{\phantom{j}}_0\mathcal{F}_0^{(1)}(s;iw)=   \frac{\det(e^{is_j w_k} )}{\det(s_j^{k-1})\det(w_j^{k-1})} .\end{equation}Thus, when $\beta=2$, the Gaussian ensemble with external source is a determinantal process.  It was extensively studied by several authors since 1996, starting  with Br\'ezin and Hikami \cite{bh0,bh1,bh4}, Zinn-Justin \cite{zinn1,zinn2}, and Bleher and Kuijslaars \cite{bk1,bk2}. In the latter references, the eigenvalue correlation functions (marginal densities) were shown to be exactly computable in terms of multiple Hermite polynomials \cite{delvaux,df1b,kuijlaars}.   A physical interpretation for this ensemble  was also proposed (see  \cite{daems} and references therein):   setting
$$\lambda_i=\sqrt{\frac{2t(1-t)}{N}}x_i, \qquad \pi_i=\sqrt{\frac{2(1-t)}{Nt}}f_i,\qquad 1\leq i\leq N,\qquad 0<t<1,$$
 the $\beta=2$ eigenvalue density, given by \eqref{eqbetadensities} and \eqref{iz},  becomes equal to the density at time $t$ for $N$ independent non-intersecting Brownian motions on the line, $(\lambda_1,\ldots \lambda_N)$, such that each $\lambda_i$ starts ($t=0$) at the origin and ends ($t=1$) at the point $\pi_i$. If we suppose that the external source $\mathbf{F}$ is only of finite rank $r$, which means
\begin{equation}\label{finiterank} f_1\neq 0,\ldots,f_r\neq 0,\qquad f_{r+1}=0,\ldots,f_{N}=0,\qquad \lim_{N\to\infty}\frac{r}{N}=0, \end{equation}
then we obtain a particularly beautiful  model of Brownian motions with a few outliers \cite{adler}.

The interest in ensembles with finite rank external source was prompted by the work of Baik, Ben Arous, and P\'ech\'e \cite{bbp}.  While analyzing the distribution function for the largest eigenvalue,  $x_1$ say, of the spiked complex Wishart model,
the authors discovered (completely independent from \cite{lang,jones}) a phase transition phenomenon and  obtained the limiting distributions for $x_1$, both at the critical point and away from the critical point.  This analysis was almost immediately adapted by P\'ech\'e  \cite{peche} to the GUE with finite rank external source, which in our notation,
is defined by  the eigenvalue density  \eqref{eqbetadensities} with $\beta=2$ and  equation \eqref{finiterank}.  It is actually more convenient to rescale the variables as follows:
\begin{equation}\label{rescaledvariables}
  x_i= \sqrt{\frac{N}{2}}\lambda_i,\qquad   f_i=\sqrt{\frac{N}{2}}\pi_i,\qquad i=1,\ldots, N.
\end{equation}For a rank $r=1$ perturbation, P\'ech\'e found   three phases for the distribution of the largest eigenvalue $\lambda_1$ in the neighborhood of the soft edge of the spectrum.  Following the nomenclature used in \cite{peche}, the phases are divided as follows:

\begin{tabular}{ll} \textbf{Subcritical regime}: & If $\pi_1<1$, then   $\displaystyle\lim_{N\to\infty} \mathrm{Prob}\left( N^{2/3}(\lambda_1-2)\leq y\right)=F_{2}(y)$.\\
 \textbf{Critical regime}:& If $\pi_1=1$, then $\displaystyle\lim_{N\to\infty} \mathrm{Prob}\left( N^{2/3}(\lambda_1-2)\leq y\right)=F_{2+1}(y)$.\\
 \textbf{Supercritical regime}:& If $\pi_1>1$,  $\displaystyle c=\pi_1+1/{\pi_1}$, and $\displaystyle \sigma^2={\pi_1^2}/{(\pi_1^2-1)}$, then \\ &
$\displaystyle \lim_{N\to\infty} \mathrm{Prob}\left(\sigma^2\sqrt{N}(\lambda_1-c)\leq y\right)=\frac{1}{2}\left(1+\mathrm{erf}\frac{y}{\sigma\sqrt{2}}\right)$.
\end{tabular}

\noindent See \cite{peche} for the definition of the distribution functions $F_{2+k}$. Ensembles of complex hermitian matrices  with finite rank external source were subsequently studied by many authors, see for instance \cite{df1,df2,adler,bw,bff, bertola,br}.

\subsection{Goals} We are interested in studying correlations of characteristic polynomials for the Gaussian $\beta$-ensemble with external source, when $\beta$ is any positive real and the finite rank condition \eqref{finiterank} is satisfied.  So far, few authors have worked on $\beta$-ensembles with external source for generic values of $\beta$. Dualities relating expectation values of products of characteristic polynomials were studied in \cite{des,for2012}.  The latter reference also contains the limiting expectation value of a single characteristic polynomial in the critical regime.   In \cite{forspiked,for2012}, different approaches used for defining $\beta$-ensembles with external source were shown to be equivalent.  Finally, the distribution of the largest eigenvalue    was studied in   \cite{bv1,bv2, wang}, where a phase transition phenomenon, completely similar to that described previously, was also revealed.

More specifically, we aim to get exact closed-form expressions for the asymptotic limit of the expectation value of  a product  of $n$ characteristic polynomials.  We moreover want to prove that the phase transition at the soft edge is observable not only for the distribution of the largest eigenvalue, as noticed in \cite{bv1}, but also at the level of correlations of characteristic polynomials.

Thus, throughout the article, we want to determine the large $N$ behavior, under the assumption of the finite rank condition \eqref{finiterank},  of the following expectation value  of products of characteristic polynomials:
\begin{multline}\label{productmean}
K_{\beta,N}(s_1,\ldots, s_n;f_1,\ldots, f_r)=\\ \frac{1}{ G_{\beta,N}}\int_{\mathbb{R}^{N}}  \prod_{j=1}^n\prod_{i=1}^N (s_j-x_i)\, e^{-\frac{\beta}{2}\sum_{i=1}^N(x_i^2+f_i^2)}
  \, |\Delta_{N}(x)|^{\beta}\!\! {\phantom{f}}_0\mathcal{F}_0^{(2/\beta)}(\beta x;f)\,d^{N}x\, ,
\end{multline}
where we have used  a shorthand notation for the Vandermonde determinant, that is,
\begin{equation}\label{vandermonde} \Delta_{N}(x)=\prod_{1\leq j<k\leq N}(x_{k}-x_{j})\, .
\end{equation}

Obviously, the integral representation \eqref{productmean}  is not suitable for studying the asymptotic limit of $K_{\beta,N}(s;f)$ as  $N\to\infty$.  However, as was shown in \cite[Proposition 7]{des}, there exists a duality formula  that provides an alternative $n$-dimensional integral representation for  \eqref{productmean}:
 \begin{equation}\label{duality}
K_{\beta, N}(s;f)\,= \, D_{\!\beta,N,n}\, e^{\sum_{j=1}^{n} s^{2}_j}\,\int_{\mathbb{R}^{n}}  \prod_{j=1}^n\prod_{k=1}^N (if_k-y_j)
e^{- \sum_{l=1}^{n} y^{2}_l}\, |\Delta_{n}(y)|^{4/\beta}\!\! {\phantom{f}}_0\mathcal{F}_0^{(\beta/2)}(y;2is)\,d^{n}y,
\end{equation}
where $ D_{\!\beta,N,n}$ denotes a constant, which  is given in \eqref{dualityconstant}.  Notice that one essentially goes from \eqref{productmean} to \eqref{duality}  by changing $(\beta,N,n,f,s)$ into $(4/\beta,n,N,s,f)$.

\subsection{Main results}

We  now give the asymptotic limits for the averaged product of $n$ characteristic polynomials.  It is actually more convenient to display the results for the weighted expectation
\be \label{weightedquantity}
\varphi_{\beta, N}(s;f)= e^{ -  \frac12 \sum_{j=1}^{n} s^{2}_j   } \, K_{\beta,N}(s;f).
\ee

\subsubsection{Multivariate functions}
The asymptotic results are written in terms of new  multivariate functions, which are  of Airy and Gaussian types.  They are defined below, but  will be further studied  in Section 2. Note that in the following definitions, it is understood that $\prod_{k=1}^{m}(i w_{j}+f_{k})=1$ whenever $m=0$.

\begin{definition} \label{defAiryfunction} For $s\in \mathbb{R}_+^n$ and $f\in \mathbb{C}^m$, the  incomplete multivariate   Airy function  is
\be\label{Airydef} \mathrm{Ai}_{n,m}^{(\alpha)}(s;f)=\frac{1}{(2\pi)^n}\int_{\mathbb{R}^n}e^{i\sum_{j=1}^{n}w_j^{3}/3}\prod_{j=1}^{n}\prod_{k=1}^{m}(i w_{j}+f_{k}) \ |\Delta_{n}(w)|^{2/\alpha}{\phantom{j}}_0\mathcal{F}_0^{(\alpha)}(s;i w)\, d^{n}w.   \ee
 For $s\in\mathbb{C}^n\setminus \mathbb{R}_+^n$, the function  $\mathrm{Ai}_{n,m}^{(\alpha)}$ is defined similarly, except that  each variable $w_j$ follows a complex path going from  $-\infty+i\delta $ to  $\infty+i\delta $ for  some $\delta>0$.
\end{definition}

It is worth mentioning that for $\alpha=1$ and $m=0$, the above function is equivalent to Kontsevich's version of the matrix Airy function \cite{kon}.  Moreover,  for all $\alpha>0$,
$$ \left.\mathrm{Ai}_{n,m}^{(\alpha)}(s;f)\right|_{m\equiv 0}= \mathrm{Ai}^{(\alpha)}(s),$$ where the function on the RHS is  the multivariate Airy  function in one set of variables, which previously appeared in \cite{des,dl}.
Asymptotic series of $\mathrm{Ai}_{n,m}^{(\alpha)}(s;f)$ as $s_j\to\pm\infty$ will be given later  in Proposition \ref{propAiry}.

\begin{definition} \label{defGaussianfunction}  For $s\in \mathbb{C}^n$ and $f\in \mathbb{C}^m$,  the   multivariate   Gaussian  function is
\be\label{Gdef} G_{n,m}^{(\alpha)}(s;f)=\frac{1}{\Gamma_{2/\alpha,n}}\int_{\mathbb{R}^n}e^{-\sum_{j=1}^{n}w_j^{2}/2}\prod_{j=1}^{n}\prod_{k=1}^{m}(i w_{j}+f_{k}) \ |\Delta_{n}(w)|^{2/\alpha}{\phantom{j}}_0\mathcal{F}_0^{(\alpha)}(s;i w)\, d^{n}w,  \ee
where the constant $\Gamma_{2/\alpha,n}$ is given in \eqref{constgamma}.
\end{definition}

Let $\beta=2/\alpha$.  Then, $G_{n,m}^{(\alpha)}(s;f)$  is proportional to the expectation of a product of $m$ characteristic polynomials for matrices of size $n$ in the Gaussian $\beta$-ensemble with external source.   As will be proved in Section 2, this multivariate Gaussian function can be written explicitly as a series involving Jack polynomials and multivariate
Hermite polynomials.

\subsubsection{Soft edge limits}

In order to get the correlations of characteristic polynomials at the soft edge of the spectrum, we have to rescale the spectral variables either as
  \be \label{rescalings} s_j=\sqrt{\frac{N}{2}}\mu+\frac{\bar s_j}{\sqrt{2N^{1/3}}},\quad j=1,\ldots,n,
	\ee
	or as
	  \be \label{rescalings2} s_j=\sqrt{\frac{N}{2}}\mu+\frac{\bar s_j}{\sqrt{2}\sigma },\quad j=1,\ldots,n.
	\ee
  Note that $\mu$ and $\sigma$ are real positive parameters. The sources must also be rescaled as
		\be \label{orderingpi} f_k=\sqrt{\frac{N}{2}}\pi_k,\qquad k=1,\ldots,N.\ee
  To avoid any confusion, we rewrite the finite rank criterion \eqref{finiterank} as follows:
	\be \label{orderingpi2} \pi_1\geq \cdots \geq \pi_r,  \qquad  \pi_{r+1}=\cdots=\pi_N=0. \ee
We  stress that the order of $\pi_j$'s is not essential and the key point is  whether  they equal   the critical values or not.

The next theorems describe the phase transition phenomenon.  If $\pi_1<1$, then the correlation $\varphi_{\beta,N}(s;f)$ is asymptotically the  same as the correlation for a $\beta$-ensemble without external source, which was obtained in \cite{dl}.   For $\pi_1=1$ and nearby, new asymptotic correlations occur, but are still of Airy type.  For  $\pi_1>1$,   the correlations become those that one would normally observe  for matrices of size $n$ in a Gaussian $4/\beta$-ensemble.  We stress that the scalings \eqref{rescalings} and the critical value, which is $\pi_1=1$,  do not depend on $\beta$, in accordance with what was found for the distribution of the largest eigenvalue \cite{bv1}.

\begin{theorem}[Subcritical regime] \label{theosoftsub}
Assume \eqref{rescalings},  \eqref{orderingpi}, \eqref{orderingpi2}, $\mu=2$, and $\pi_1<1$.  Then, as $N\to\infty$,
\be    \frac{\varphi_{\beta,N}(s_1,\ldots,s_n;f_1,\ldots,f_r)}{\Phi_{\!^{\beta,N,n}}}  \,\sim\,\prod_{k=1}^r(1-\pi_k)^n\, \mathrm{Ai}^{(\beta/2)}(\bar s_1,\ldots,\bar s_n).
\ee
The constant $\Phi_{\!^{\beta,N,n}}$ is given in   Appendix \ref{appendix}.
\end{theorem}

\begin{theorem}[Critical regime] \label{theosoft}
Assume \eqref{rescalings},  \eqref{orderingpi}, \eqref{orderingpi2}, and $\mu=2$.  Suppose moreover that
$$\pi_1=1+\frac{\bar\pi_1}{N^{1/3}},\quad  \pi_2=1+\frac{\bar\pi_2}{N^{1/3}},\quad  \ldots,  \quad  \pi_m=1+\frac{\bar \pi_m}{N^{1/3}},$$ while   $\pi_k$ belongs to a compact subset of $(-\infty,1)$ for all $m+1\leq k\leq r$.    Then, as $N\to\infty$,
\be    \frac{\varphi_{\beta,N}(s_1,\ldots,s_n;f_1,\ldots,f_r)}{\Phi_{\!^{\beta,N,n,m}}}  \,\sim\,\prod_{k=m+1}^r(1-\pi_k)^n\, \mathrm{Ai}_{n,m}^{(\beta/2)}(\bar s_1,\ldots,\bar s_n; \bar \pi_1,\ldots,\bar \pi_m).
\ee
The constant $\Phi_{\!^{\beta,N,n,m}}$ is given in  Appendix \ref{appendix}.
\end{theorem}

\begin{theorem}[Supercritical regime] \label{theosoftsup} Assume \eqref{rescalings2},  \eqref{orderingpi}, \eqref{orderingpi2}, and $\mu>2$. Let $\sigma^2=\nu^2/(\nu^2-1)$, where $\nu>1$ is such that $\mu=\nu+\nu^{-1}$. Suppose moreover that
$$\pi_1=\nu+\frac{\sigma\bar\pi_1}{N^{1/2}},\quad \pi_2=\nu+\frac{\sigma \bar\pi_2}{N^{1/2}}, \quad   \ldots, \quad   \pi_m=\nu+\frac{\sigma \bar \pi_m}{N^{1/2}}, $$ while   $\pi_k$ belongs to a compact subset of $(-\infty,\nu)$ for all $m+1\leq k\leq r$.    Then, as $N\to\infty$,
\begin{multline}  \prod_{j=1}^n e^{\frac{2\nu-\mu}{2\sigma}\sqrt{N} \bar{s}_{j}}\  \frac{\varphi_{\beta,N}(s_1,\ldots,s_n;f_1,\ldots,f_r)}{\Phi^\mathrm{sup}_{\!^{\beta,N,n,m}}}  \\
\,\sim\,\prod_{k=m+1}^r(\nu-\pi_k)^n\, \prod_{j=1}^ne^{\frac{1}{4\sigma^2}\bar s_j^2}\,G_{n,m}^{(\beta/2)}(\bar s_1,\ldots,\bar s_n; \bar \pi_1,\ldots,\bar \pi_m).\label{theosoftsupeq}
\end{multline}
The constant $\Phi^\mathrm{sup}_{\!^{\beta,N,n,m}}$ is given in  Appendix \ref{appendix}.
\end{theorem}

\subsubsection{Bulk limits} Different scalings are required for this part of the spectrum.  For   $u\in (-1,1)$, $1\leq j \leq n$, and $1\leq k \leq r $,  let
  \be \label{rescalingsbulk}  s_{j}=\sqrt{2N}u+ \frac{\pi \bar s_j}{\sqrt{2N(1-u^{2})}} \qquad\text{and}\qquad f_{k}=\sqrt{\frac{N}{2}} \,(u+\sqrt{1-u^{2}} \pi_{k}) .\ee
The scaling correlations in the bulk are given below.  In contradistinction with the soft edge case, the limit correlations in the bulk   can be continuously deformed  into that for the ensemble without source \cite{dl} -- here corresponding to $\pi_j=0, j=1, \dots, r$ -- so no phase transition has been found in this part of the spectrum.



\begin{theorem}[Bulk limit] \label{theobulk} Assume \eqref{rescalingsbulk}  and \eqref{orderingpi2}. Then as $N\to \infty$
\begin{equation}\label{theobulkeq1} (\Psi_{\!^{N,2m}})^{-1} \varphi_{\beta,N} (s;f)\,\sim\,
\gamma_{m}\!(\!4\!/\!\beta\!)\prod_{k=1}^{r}(1+\pi_{k}^{2})^{m}\, e^{-i\pi \sum_{j=1}^{n}\bar s_{j}}\! {\phantom{j}}_1F_1^{(\beta/2)}(2m/\beta;2n/\beta;2i\pi \bar s)\end{equation}
for $n=2m$  while
\begin{align} \label{theobulkeq2}\frac{1}{\Psi^{\!_{(0)}}_{\!^{N,2m-1}}\Psi^{\!_{(1)}}_{\!^{N,2m-1}}}\Big\{\varphi_{\beta,N} (s;f)\varphi_{\beta,N-1} (s';f)-\varphi_{\beta,N} (s';f)\varphi_{\beta,N-1} (s;f)\Big\}\sim \prod_{k=1}^{r}(1+\pi_{k}^{2})^{2m-1}\nonumber\\\frac{1}{2i}\Big\{ e^{i\pi \sum_{j=1}^{n}(\bar s'_{j}-\bar s_{j})} {\phantom{j}}_1F_1^{(\beta/2)}(2m/\beta;2n/\beta;2i\pi \bar s){\phantom{j}}_1F_1^{(\beta/2)}(2m/\beta;2n/\beta;-2i\pi \bar s')-(\bar s\leftrightarrow \bar s')\Big\}\end{align}
for $n=2m-1$. The coefficients  are  given in  Appendix \ref{appendix}.
\end{theorem}

\begin{remark}[Slowly growing rank case] \label{growingrank} When the rank $r=r_N$  depends on $N$, but grows slowly with $N$,  our main results remain true. To be more precise, let us replace the finite rank condition \eqref{orderingpi2} by
\be \label{orderingpiinfiniterank} \pi_1\geq \cdots \geq \pi_r,  \qquad  \pi_{r+1}=\cdots=\pi_N=0, \qquad  \lim_{N\to\infty} \frac{ r}{ N^{b}}=0\  \mathrm{for \ some \ b\in (0,1]}. \ee
We suppose moreover that the integer $m$ of  Theorems \ref{theosoft} and \ref{theosoftsup} is finite, and  also that each $\pi_k$ belongs to a compact subset of $(-\infty,1)$ (resp.\ $\mathbb{R}$) for all $1\leq k\leq r$ in Theorem  \ref{theosoftsub} (resp.\ Theorem \ref{theobulk}). Then, Theorems~\ref{theosoftsub}~to~\ref{theobulk}
 hold respectively for $b=1/3,\,  1/3, \, 1/2, \, 1/2$. See Section \ref{sectionslow} for more details.  This kind of slowly growing rank perturbation was studied by P\'{e}ch\'{e}  for the largest eigenvalue of the Gaussian Unitary Ensemble \cite{peche}.
\end{remark}

The rest of the paper is devoted to the study of the multivariate functions in Section 2, followed by the proof of Theorems \ref{theosoftsub}, \ref{theosoft}, \ref{theosoftsup},  and \ref{theobulk}.


\section{Jack polynomials and hypergeometric functions}
\label{hypergeometric}
This section first provides a brief review of  Jack polynomial theory  
  and the associated multivariate hypergeometric functions.
  A recent textbook presentation of these multivariate functions  can be found in Chapters 12 and 13 of Forrester's book \cite{forrester}; a classical reference for the Jack polynomials is Macdonald's book \cite[Chap. VI 10]{macdonald}.  A few results stated here have been proved in the previous paper \cite{dl}  and  will be used later.  We finally derive new results on the multivariate functions of Gaussian type and of Airy type.

\subsection{Partitions}

A partition $\kappa = (\kappa_1, \ldots,\kappa_i,\ldots)$ is a sequence of non-negative integers $\kappa_i$ such that
\begin{equation*}
    \kappa_1\geq\kappa_2\geq\cdots\geq\kappa_i\geq\cdots
\end{equation*}
and only a finite number of the terms $\kappa_i$ are non-zero. The number of non-zero terms is referred to as the length of $\kappa$, and is denoted by $\ell(\kappa)$. We shall not distinguish between two partitions that differ only by a string of zeros. The weight of a partition $\kappa$ is the sum
\begin{equation*}
    |\kappa|:= \kappa_1+\kappa_2+\cdots
\end{equation*}
of its parts, and its diagram is the set of points $(i,j)\in\mathbb{N}^2$ such that $1\leq j\leq\kappa_i$.  Reflection with respect to the diagonal produces the conjugate partition
$\kappa^\prime=(\kappa_1',\kappa_2',\ldots)$.

The set of all partitions of a given weight is partially ordered
by the dominance order: $\kappa\leq \sigma $ if and only if $\sum_{i=1}^k\kappa_i\leq \sum_{i=1}^k \sigma_i$ for all $k$.

\subsection{Jack polynomials}

Let $\Lambda_n(x)$ denote the algebra of symmetric polynomials in $n$ variables $x_1,\ldots,x_n$ and with coefficients in the field $\mathbb{F}$.    In this article,   $\mathbb{F}$    is assumed to be the field of  rational functions in the parameter $\alpha$.  As a ring, $\Lambda_n(x)$ is generated by the power-sums:
\be \label{powersums} p_k(x):=x_1^k+\cdots+x_n^k. \ee
The ring of symmetric polynomials is naturally graded: $\Lambda_n(x)=\oplus_{k\geq 0}\Lambda^k_n(x)$, where $\Lambda^k_n(x)$ denotes the set of homogeneous polynomials of degree $k$.   As a vector space, $\Lambda^{k}_n(x)$ is equal to the span over  $\mathbb{F}$ of all symmetric monomials $m_\kappa(x)$, where $\kappa$ is a partition of weight $k$ and
\be  m_\kappa(x):=x_1^{\kappa_1}\cdots x_n^{\kappa_n}+\text{distinct permutations}.\nonumber
\ee
Note that if the length of the partition $\kappa$ is larger than $n$,  we set $m_\kappa(x)=0$.

The whole ring $\Lambda_n(x)$ is invariant under the action of  homogeneous differential operators related to the Calogero-Sutherland models \cite{bf}:
\be E_k=\sum_{i=1}^n x_i^k\frac{\partial}{\partial x_i},\qquad D_k=\sum_{i=1}^n x_i^k\frac{\partial^2}{\partial x_i^2}+\frac{2}{\alpha}\sum_{1\leq i\neq j \leq n}\frac{x_i^k}{x_i-x_j}\frac{\partial}{\partial x_i},\qquad k\geq 0.\label{Calogero}
\ee
The operators $E_1$ and $D_2$ are special since they also preserve each $\Lambda^k_n(x)$.  They can be used to define the Jack polynomials.  Indeed, for each partition $\kappa$, there exists a unique symmetric polynomial $P^{(\alpha)}_\kappa(x)$ that satisfies the following two  conditions: 
\begin{align} \label{Jacktriang}(1)\qquad&P^{(\alpha)}_\kappa(x)=m_\kappa(x)+\sum_{\mu<\kappa}c_{\kappa\mu}m_\mu(x)&\text{(triangularity)}\\
\label{Jackeigen}(2)\qquad &\Big(D_2-\frac{2}{\alpha}(n-1)E_1\Big)P^{(\alpha)}_\kappa(x)=\epsilon_\kappa P^{(\alpha)}_\kappa(x)&\text{(eigenfunction)}\end{align}
where the coefficients $\epsilon_\kappa$ and  $c_{\kappa\mu}$ belong to  $\mathbb{F}$.  Because of the triangularity condition, $\Lambda_n(x)$ is also equal to the span over $\mathbb{F}$ of all Jack polynomials  $P^{(\alpha)}_\kappa(x)$, with $\kappa$ a partition of length less than or equal to $n$.

\subsection{Hypergeometric series}\label{subsecseries}

Recall that the arm-lengths and leg-lengths of the box $(i,j)$ in the partition $\kappa$  are respectively given by
\begin{equation}\label{lengths}
    a_\kappa(i,j) = \kappa_i-j\qquad\text{and}\qquad l_\kappa(i,j) = \kappa^\prime_j-i.
\end{equation}We define the hook-length of a partition $\kappa$  as the following product:
\begin{equation} \label{defhook}
    h_{\kappa}^{(\alpha)}=\prod_{(i,j)\in\kappa}\Big(1+a_\kappa(i,j)+\frac{1}{\alpha}l_\kappa(i,j)\Big),
\end{equation}
and closely related is the following $\alpha$-deformation of the Pochhammer symbol:
\begin{equation}\label{defpochhammer}
    [x]^{(\alpha)}_\kappa = \prod_{1\leq i\leq \ell(\kappa)}\Big(x-\frac{i-1}{\alpha}\Big)_{\kappa_i} = \prod_{(i,j)\in\kappa}\Big(x+a^\prime_\kappa(i,j)-\frac{1}{\alpha}l^\prime_\kappa(i,j)\Big).
\end{equation}
In the middle of the last equation,  $(x)_j\equiv x(x+1)\cdots(x+j-1)$ stands for the ordinary Pochhammer symbol, to which $\lbrack x\rbrack^{(\alpha)}_\kappa$ clearly reduces for $\ell(\kappa)=1$.  The right-hand side of \eqref{defpochhammer} involves the co-arm-lengths and co-leg-lengths box $(i,j)$ in the partition $\kappa$, which are respectively defined as
\begin{equation}\label{colengths}
    a^\prime_\kappa(i,j) = j-1,\qquad \text{and}\qquad l^\prime_\kappa(i,j) = i-1.
\end{equation}

We are now ready to give the  precise definition of the hypergeometric series used in the article. 
 Fix $p,q\in\mathbb{N}_0=\{0, 1, 2,\ldots\}$ and let $a_1,\ldots,a_p, b_1,\ldots,b_q$ be complex numbers such that $(i-1)/\alpha-b_j\notin\mathbb{N}_0$ for all $i\in\mathbb{N}_0$.
The $(p,q)$-type hypergeometric series refers to
\begin{equation}\label{hfpq}
  {\phantom{j}}_p F^{(\alpha)}_{q}(a_1,\ldots,a_p;b_1,\ldots,b_q;x) = \sum_{k=0}^{\infty}\sum_{|\kappa|=k} \frac{1}{h_{\kappa}^{(\alpha)}}\frac{\lbrack a_1\rbrack^{(\alpha)}_\kappa\cdots\lbrack a_p\rbrack^{(\alpha)}_\kappa}{\lbrack b_1\rbrack^{(\alpha)}_\kappa
    \cdots\lbrack b_q\rbrack^{(\alpha)}_\kappa}P_{\kappa}^{(\alpha)}(x).
\end{equation}
Similarly, the hypergeometric series in two sets of $n$ variables  $x=(x_1,\ldots,x_n)$ and  $y=(y_1,\ldots,y_n)$  is given   by
\begin{equation}\label{hfpq2}
    {\phantom{j}}_p\mathcal{F}^{(\alpha)}_{q}(a_1,\ldots,a_p;b_1,\ldots,b_q;x;y)
                = \sum_{\kappa} \frac{1}{h_{\kappa}^{(\alpha)}}\frac{\lbrack a_1\rbrack^{(\alpha)}_\kappa\cdots\lbrack a_p\rbrack^{(\alpha)}_\kappa}{\lbrack b_1\rbrack^{(\alpha)}_\kappa
    \cdots\lbrack b_q\rbrack^{(\alpha)}_\kappa}\frac{P_{\kappa}^{(\alpha)}(x)P_{\kappa}^{(\alpha)}(y)}{P_{\kappa}^{(\alpha)}(1^{n})},
\end{equation}
        where we have used the shorthand notation $1^{n}$ for $\overbrace{1,\ldots,1}^{n}$.
Note that when $p\leq q$, the above series converge absolutely for all $x\in\mathbb{C}^n$,  $y\in\mathbb{C}^n$ and  $\alpha\in\mathbb{R}_+$.  In the case where $p=q+1$, then the series converge absolutely for all $| x_i|<1$, $| y_i|<1$ and  $\alpha\in\mathbb{R}_+$.  



Now we give    a translation  property of ${\phantom{j}}_0\mathcal{F}^{(\alpha)}_0$, which proves  to be of practical importance.
  For convenience, we write
        \be (a^n)=(\overbrace{a,\ldots,a}^{n}),\qquad b+a x=(b+a x_{1},  \ldots, b+a x_{n}),\nonumber\ee
where $a, b$ are complex numbers  and $x=(x_{1},  \ldots, x_{n})$. Then   we have
  \be \label{vip1}{\phantom{j}}_0\mathcal{F}^{(\alpha)}_0(a+x;b+y)=\exp\{nab+a p_{1}(y)+b p_{1}(x)\}\,{\phantom{j}}_0\mathcal{F}^{(\alpha)}_0(x;y) \ee
and
  \begin{align}
{\phantom{j}}_0 \mathcal{F}_{0}^{(\alpha)}&(x_1,\ldots,x_n;a^k, b^{n-k})=e^{b p_1(x)} {\phantom{j}}_1 F^{(\alpha)}_{1}(k/\alpha; n/\alpha; (a-b)x_1,\ldots,(a-b)x_n)\label{vip2}\\ \displaystyle
&=e^{(a-b)kx_1+b p_1(x)} {\phantom{j}}_1 F^{(\alpha)}_{1}(k/\alpha; n/\alpha; (a-b)(x_2-x_1),\ldots,(a-b)(x_n-x_1)).\label{vip3}
\end{align}
See \cite{dl} for the proofs of \eqref{vip1}--\eqref{vip3}.

\subsection{Airy type functions}

The incomplete multivariate   Airy function $\mathrm{Ai}_{n,m}^{(\alpha)}$ has been introduced in Definition \ref{defAiryfunction} .  In the  case of $m=0$ and general $\alpha$, this definition coincides with that of the multivariate Airy function of \cite{des,dl}. For the special values of $\alpha=2, 1, 1/2$, the RHS of   \eqref{Airydef} is proportional to the following matrix Airy function:
\be \int \exp\{i \TR(\frac{1}{3} W^{3}+  SW)\}\,\prod_{j=1}^{m}\det(iW+f_{j})\, d W, \label{matrixairy}\ee
where both $S$ and $W$ are either symmetric ($\alpha=2$), Hermitian ($\alpha=1$) or self-dual quaternion ($\alpha=1/2$) $n \times n$ matrices. Actually, in  order to ensure  convergence of \eqref{matrixairy}, one  first assumes that $S$ is  positive definite.  The latter matrix can be  further restricted  to be of diagonal form because of the orthogonal, unitary or symplectic conjugate invariance for the integral.  One then extends  the integral  to the general case by imposing that the matrix entries of $W$ to follow some contour in the complex plane.
When  $m=0$ and   $\alpha=1$, the above integral definition was first given by Kontsevich \cite{kon}.

\begin{proposition}[Closed-form expression for $\alpha=1$] Define the operators $L_{k}$, for $k=1,2,\ldots$, as
\be\label{opLk} L_{k}\,h(x)=(\frac{d}{d x})^{k-1}\prod_{l=1}^{m}\big(\frac{d}{d x}+f_l \big)\,h(x), \ee
for some smooth function $h(x)$.  Then
\be \label{Airy1}\mathrm{Ai}_{n,m}^{(1)}(s;f)=\prod_{k=1}^{n}k!  \  \frac{\det(L_{k}\mathrm{Ai}(s_j))}{\det( s_j^{k-1})}.\ee
\end{proposition}
\begin{proof}
The formula of the  Harish-Chandra-Itzykson-Zuber integral gives
$${\phantom{j}}_0\mathcal{F}_0^{(1)}(s;iw)=  c \frac{\det(e^{is_j w_k} )}{\det(s_j^{k-1})\det(w_j^{k-1})}
$$ for some constant $c$,  see Section 4 of \cite{kon}. We next determine $c$ by expanding the determinant $\det(e^{is_j w_k} )$ for the series
$$e^{is_j w_k}=\sum_{l\geq 0} \frac{i^{l}}{l!}(s_j w_k)^{l},$$
 then set $s_j=w_j=0$, and get $c=i^{-n(n-1)/2}\prod_{j=1}^{n-1}j!$.  
 Simple manipulations  give  the desired result \eqref{Airy1}.
\end{proof}

We conclude this subsection with a proposition  devoted to the asymptotic behavior of the incomplete multivariate   Airy function. The proof   will be given in Section \ref{rmt}. Notice   the following shorthand notation:  $(A+Bs)$ stands for $(A+Bs_1,\ldots, A+B s_n)$ for $A, B\in \mathbb{C}$.

\begin{proposition}\label{propAiry}  As the real positive variable $x\to \infty$, the following hold.

 \noindent $\mathrm{(i)}$  For some $0\leq k\leq r$, let
 $$\bar f_{j}=1+(2x^{3/2})^{-1/2} f_j\ (1\leq j \leq k)\  \mathrm{and} \  \bar f_{l}=  f_l\neq 1 \ (k< l \leq r), $$
 then \begin{equation}
\mathrm{Ai}_{n,r}^{(\alpha)}(x+x^{-1/2}s;x^{1/2}\bar f)\sim\frac{\Gamma_{2/\alpha,n}\,\prod_{l=k+1}^{r}(f_{l}-1)^{n}}{(2\pi)^{n}2^{((1+k)n+n(n-1)/\alpha)/2}}
\frac{e^{-\frac{2n}{3}x^{3/2}}e^{-\sum_{j=1}^{n}s_{j}} }{x^{((1-2r+3k)n+n(n-1)/\alpha)/4}}\,G_{n,k}^{(\alpha)}(0;f).  \label{airy1}\end{equation}

\noindent $\mathrm{(ii)}$ For $n=2m$,
\begin{align}
\mathrm{Ai}_{n,r}^{(\alpha)}(-x+x^{-1/2}s;x^{1/2}  f)&\sim (2\pi)^{-n} \tbinom{2m}{m} (\Gamma_{2/\alpha,m})^2(2\sqrt{x})^{-m+m(m+1)/\alpha} x^{rm}
  \nonumber\\
 &\times \prod_{l= 1}^{r}(1+f_{l}^{2})^{m} e^{-i \sum_{j=1}^{n}s_{j}} \!\!{\phantom{j}}_1F_1^{(\alpha)}(m/\alpha;n/\alpha;2is)\label{airy2}.  \end{align}
\end{proposition}

\begin{remark} Another type of multivariate  function defined by
\be\label{Airydefinverse}  \frac{1}{(2\pi)^n}\int_{\mathbb{R}^n}e^{i\sum_{j=1}^{n}w_j^{3}/3}\prod_{j=1}^{n}\prod_{k=1}^{m}(i w_{j}+f_{k})^{-1/\alpha} \ |\Delta_{n}(w)|^{2/\alpha}{\phantom{j}}_0\mathcal{F}_0^{(\alpha)}(s;i w)\, d^{n}w  \ee
  may also be useful in Random Matrix Theory, where all $f_{k}$ lie  in the domain $\{z\in \mathbb{C}:\textrm{Re}\,z\neq 0\}$, although it is not used in the present paper.
\end{remark}

\subsection{Gaussian type functions} We have introduced the function $G_{n,m}^{(\alpha)}$   in Definition  \ref{defGaussianfunction}.
Obviously, for the cases $\alpha=2,1,1/2$, this function can be interpreted as the expectation of a product of $m$ characteristic polynomials:
\be G_{n,m}^{(\alpha)}(s;f)=c_{\beta,n}\int \exp\{- \frac{1}{2}\TR W^{2}+\TR  SW\}\,\prod_{j=1}^{m}\det(iW+f_{j})\, d W, \ee
where $c_{\beta,n}$ is a normalization constant while both $S$ and $W$ are either symmetric ($\alpha=2$), Hermitian ($\alpha=1$) or self-dual quaternion ($\alpha=1/2$) $n \times n$ matrices.

Similarly  to the Airy functions, there are determinantal formulas in the unitary ($\alpha=1$) case.
\begin{proposition}[Closed-form expression for $\alpha=1$] Let $L_{k}$ be the operator in \eqref{opLk}.   Then,
\be \label{G1}G_{n,m}^{(1)}(s;f)=  \frac{
\det(L_{k}e^{-s^{2}_{j}/2})
}{\det( s_j^{k-1})}.  \ee
\end{proposition}

Other explicit formulas are given below.
\begin{proposition}[Closed-form expression for $m=1$]   Let $f_1=z$, $a=\sqrt{2/\alpha}$ and $H_k(z)=\sum_{l=0}^{\lfloor k/2\rfloor}c_{c-l}z^{k-2l}$ denote the standard Hermite polynomial of degree $k$.  Define $\bar H_k(z)=\sum_{l=0}^{\lfloor k/2\rfloor}(-1)^l c_{c-l}z^{k-2l}$.  Then
$$ G_{n,1}^{(\alpha)}(a\,s ; a\,z)=\frac{e^{-p_2(s)/\alpha}}{(2\alpha)^{n/2}} \sum_{k=0}^n(-2)^{n-k}e_{n-k}(s)\bar H_k(z),$$
where $e_k(s)=\sum_{1\leq i_1<\cdots<i_k\leq n}s_{i_1}\cdots s_{i_k}$ denotes the elementary symmetric function of degree $k$.
\end{proposition}
\begin{proof}
This directly  follows from:
(1) the  duality formula \eqref{duality}, which gives
\be\label{Gn1} G_{n,1}^{(\alpha)}(s;z)=\frac{1}{\sqrt{2\pi}}e^{-\sum_{j=1}^{n}s_j^{2}/2}\int_{-\infty}^{\infty}e^{-y^{2}/2}\prod_{j=1}^{n}(y/\sqrt{\alpha}+ z-s_{j}) \, d y; \ee
(2) the generating function $\prod_{j=1}^n(z+s_j)=\sum_{k=0}^n z^{k}e_{n-k}(s)$; and
(3) the following integral representation of the Hermite polynomials:
$$ H_n(z)=\frac{2^n}{\sqrt{\pi}}\int_\mathbb{R}(z+iv)^ne^{-v^2}dv.$$
\end{proof}

\begin{proposition}[Series expansion: general case] Let $D_0(s)$ denote the operator defined in \eqref{Calogero}, but this time for the set of variables $s$.    Then
$$G_{n,m}^{(\alpha)}(s;f)=e^{-\frac12p_2(s)}e^{-\frac12 D_0(s)}\prod_{j=1}^n\prod_{k=1}^m(f_k-s_j).$$
Equivalently, if  $H_\lambda^{(\alpha)}(s)$ is the multivariate Hermite polynomial (with monic normalization), then
$$ G_{n,m}^{(\alpha)}(\sqrt{2}s;\sqrt{2}f)=2^{nm}e^{-p_2(s)}\sum_{\lambda}H^{(\alpha)}_\lambda(s)P^{(\alpha)}_{(n^m)-\lambda'}(f).$$
\end{proposition}
\begin{proof}The first equation is a consequence of the following generalized Fourier transform, which is valid for any holomorphic function $f:\mathbb{C}^n\to\mathbb{C}$  \cite{bf}:
$$\frac{1}{\Gamma_{2/\alpha,n}}\int_{\mathbb{R}^n}e^{-\frac12 \sum_{j=1}^ny_j^2}f(iy)\,|\Delta_n(y)|^{2/\alpha} {\phantom{f}}_0\mathcal{F}_0^{(\alpha)}(-i y;s)\,d^{n}y\, =e^{-\frac{p_2(s)}{2}}e^{-\frac12 D_0(s)}f(s).$$
The second equation  follows from the first, the series expansion $$\prod_{i,j}(1+x_iy_j)=\sum_{\lambda}P_\lambda^{(\alpha)}(x)P_{\lambda'}^{(1/\alpha)}(y),$$ and Lassalle's formula (see \cite{bf})
$$ H_\lambda^{(\alpha)}(s) =e^{-\frac14 D_0(s)} P_\lambda^{(\alpha)}(s).$$
\end{proof}


\section{Proofs for the scaling limits}
\label{rmt}
In this section,  we are going to prove Theorems \ref{theosoftsub}, \ref{theosoft}, \ref{theosoftsup}, \ref{theobulk} and Proposition \ref{propAiry} through   some detailed  computations   based on   Corollaries 3.11 and 3.12 of \cite{dl}.

\subsection{General procedure}
We now want to asymptotically evaluate the weighted expectation $\varphi_{\beta, N}(s;f)$, defined in \eqref{weightedquantity}, for the finite rank perturbation case, which means $f_{r+1}=\cdots=f_N=0$.    For this, we first use the duality formula  \eqref{duality} and  introduce the scaled variables $y_{j}=\sqrt{2N}t_{j}$ on its RHS.  We also introduce a spectral parameter $u$ that  allows  us to select the part of  the spectrum we are going to study.   This allows us to rescale the spectral variables $s$ and $f$ as follows:
\be \label{scaleproof} s_j=\sqrt{2N}\Big(u+\frac{\bar s_j}{\rho N}\Big) ,\qquad  f_k=\sqrt{2N}\bar f_k , \quad\text{for all}\quad  j=1,\ldots,n,\quad k=1,\ldots,r,\ee  where $\rho$ denotes a real parameter whose value  will depend on the spectral parameter $u$. Given that the spectrum of the Gaussian ensemble is symmetrical, we restrict ourselves to $u\geq 0$.
  We then apply  \eqref{vip1} and get, for $\beta'=4/\beta$,
 \be   \varphi_{\beta,N-l} \big(s;f\big) =
 (-i\sqrt{2N})^{n(N-l)}(2\sqrt{N})^{\beta^{'}\!n(n-1)/2+n}(\Gamma_{\beta^{'}\!,n})^{-1}e^{nN u^{2}+p_2(\bar s)/(\rho^{2}N)}I_{N}(\bar s;\bar f),\label{HE1}\ee
where \begin{align} I_{N}(\bar s;\bar f)&=\int_{\mathbb{R}^n} \exp\{-N\sum\limits_{j=1}^n p(t_j)\}|\Delta_{n}(t)|^{\beta'\!}  Q(t) d^{n}t\label{HE2}\end{align}
  and
\be p(t_j)=2t_j^{2}-4iu t_j-\ln t_j, \qquad Q(t)=\prod_{j=1}^n\prod_{k=1}^r (t_j-i \bar f_{k})\prod_{j=1}^n t_j^{-l-r}\, {\phantom{j}}_{0}\mathcal{F}_0^{(2/\beta'\!)}(4i\bar s/\rho;  t-iu/2)\label{HE3}.\ee
 In the last equations, $l$ denotes a fixed non-negative integer.


  In order to evaluate $I_{N}$ as $N\to \infty$,  we will make use of the results obtained in \cite{dl} which are based on the steepest descent method.  We recall  that according to the latter method, when considering a single integration over a complex variable $z$, one first finds complex numbers $z_0$ satisfying $p'(z_0)=0$.  If $p^{(j)}(z_0)=0$ for all $j=1,\ldots, d-1$ but $p^{(d)}(z_0)\neq 0$, then we say that $z_0$ is a saddle point of degree $d-1$.  In a second time, one checks if the original path of integration (in our case, the real line) can be deformed into the path  of steepest descent, which must pass through the saddle point  $z_0$ and be such that the phase of $ \{(z-z_0)^dp^{(d)}(z_0)\}$ is zero.    In our case, since $$p'(z)=4z-4iu -1/z,$$ there are at most two saddle points $z_\pm$, which satisfy $$z_\pm=(iu\pm\sqrt{1-u^{2}})/2.$$  The nature of the saddle points depends on the value of $u$.  We distinguish three cases:
\begin{enumerate}
\item  Two complex saddle points of degree one when $u\in [0,1)$: $$z_+=   (iu+\sqrt{1-u^{2}})/2,\quad  z_-= (iu-\sqrt{1-u^{2}})/2.$$
\item  One  imaginary saddle point of degree two when $u=1$: $$z_0=i/2.$$
\item  Two  imaginary saddle points of degree one but only one accessible when $u\in (1,\infty)$:
$$z_+=   i(u+\sqrt{ u^{2}-1})/2,\quad  z_-= i(u-\sqrt{u^{2}-1})/2.$$
\end{enumerate}
\label{page3cases}

For convenience, we  list  Corollaries 3.11 and 3.12 of \cite{dl} as the following propositions.   Note that the assumptions mentioned below
  have all been verified in Subsection 4.1 of \cite{dl}  for the special case of   integral \eqref{HE2}.  Let us mention the most significant assumptions: $p(z)$ and $q(t)$ are be   analytic in some domains $\textbf{T}\subseteq \mathbb{C}$ and $\textbf{T}^{n}$, respectively; the saddle points $z_0$ and $z_\pm$ belong to   $\textbf{T}$;  the integration path along the real axis can be deformed into a path that contains straight lines passing  through the saddle points and such that $\mathrm{Re}\{p(z)-p(z_{0})\}>0$ or $\mathrm{Re}\{p(z)-p(z_{\pm})\}>0$ along these  straight lines, except possibly at the saddle points.

\begin{proposition}\label{onesaddle} Under the  assumptions (i)-(v) of Section 3.2 and (iv) of Section 3.4 in \cite{dl},  let
 \be\nonumber I_{N}=\int_{(\!a,b)^{n}}
  \exp \{-N\sum\limits_{j=1}^{n}p(t_{j}) \}|\Delta_{n}(t)|^\beta\, q(t)\,g(N^{1/d}(t-t_0))\,d^{n} t
        \ee
        where $g(t)$ is analytic in $\mathbb{C}^n$, $p(z)$ admits one  saddle point $z_0$ of order $d-1$, and $t_0=(z_0,\ldots,z_0)$.
 Then, as $N\to\infty$,
\begin{equation*} I_{N}\sim \frac{e^{-nNp(z_0)}}{N^{(n_\beta+n)/d}}  A_0\,q(t_0)  \end{equation*}
where $n_\beta=n(n-1)\beta/2$ and
\begin{equation*} A_0=\int_{\mathbb{R}^{n}} \exp\Big\{-\frac{p^{(d)}(z_0)}{d!}\sum_{j=1}^n w_j^d\Big\}\, g(w)\,|\Delta_{n}(w)|^\beta\, d^{n}w.
  \end{equation*}
 \end{proposition}

\begin{proposition}\label{twosaddle} Under the  assumptions (i)-(v) of Section 3.3 and (iv) of Section 3.4 in \cite{dl},  let
 \be\nonumber I_{N,n}=\int_{(\!a,b)^{n}}
  \exp\{-N\sum\limits_{j=1}^{n}p(t_{j})\}|\Delta_{n}(t)|^\beta\, q(t)\,d^{n} t
        \ee
        where $p(z)$ admits two simple saddle points $z_+, z_-$, and ${\rm{Re}}\{z_+-z_-\}\geq 0$.  Moreover, let $p_\pm=p''(z_\pm)$ and $\Gamma_{\beta,m}$ be given in Appendix \ref{appendix}.  If ${\rm{Re}}\{{p(z_+)}\}={\rm{Re}}\{{p(z_-)}\}$, then  as $N\to\infty$,
$$I_{N,2m}\sim\tbinom{2m}{m}(\Gamma_{\beta,m})^2\frac{(z_+-z_-)^{\beta m^2}}{(\sqrt{p_+p_-})^{m+\beta m(m-1)/2}}\frac{e^{-mN(p(z_+)+p(z_-))}}{N^{m+\beta m(m-1)/2}}q(z_+^{m},z_-^{m})$$
while
\begin{multline*} I_{N,2m-1}\sim \tbinom{2m-1}{m}\Gamma_{\beta,m-1}\Gamma_{\beta,m}
\frac{(z_+-z_-)^{\beta m(m-1)}}{(\sqrt{p_+p_-})^{m+\beta m(m-1)/2}}\frac{e^{-m N(p(z_+)+p(z_-))}}{N^{(2m-1 + \beta (m-1)^{2})/2}}\\
\times \Big(e^{Np(z_+)}(\sqrt{p_+})^{1+\beta(m-1)}q(z_+^{m-1},z_-^{m})+e^{Np(z_-)}(\sqrt{p_-})^{1+\beta(m-1)}q(z_+^{m},z_-^{m-1})\Big).  \end{multline*}
 \end{proposition}

\subsection{Bulk limit}

 \begin{proof}[Proof of  Theorem \ref{theobulk}]
We will use Proposition \ref{twosaddle} in order to establish the $N\to\infty$ asymptotic limit of the integral \eqref{HE2} in the bulk regime, which  is given by the first of the three cases enumerated on page~\pageref{page3cases}.  The method used here  is almost identical to bulk regime of the Gaussian $\beta$-ensemble without external source, which was analyzed in Section 4.1 of \cite{dl}.  To avoid repetition, we will only sketch the proof.

First, we go back to the scaling \eqref{scaleproof} and set
\be \label{choicebulk} u\in[0,1),\qquad \rho=\frac2{\pi}\sqrt{1-u^{2}},\qquad\text{and}\qquad \bar f_k=\frac{u}{2}+\frac{\sqrt{1-u^{2}}}{2}\pi_k\, .\ee
We stress that the bulk scaling used in  the Introduction, which is given in \eqref{rescalingsbulk},   immediately follows from the substitution of \eqref{choicebulk} into \eqref{scaleproof}.

Secondly, we recall that in the bulk regime, we are giving two points, $z_\pm=   (iu\pm\sqrt{1-u^{2}})/2$,  such that $p'(z_\pm)=0$ and $p''(z_\pm)\neq 0$.  Since  $u\in[0,1)$, we can write  $$ u=\sin\theta,\qquad \theta\in(-\pi/2,\pi/2),$$ which implies that
\be z_+= e^{i\theta}/2\quad \text{and}\quad z_-= e^{i(\pi-\theta)}/2.\label{eqbulkz}\ee Some simple manipulations then lead to the following equations:
\begin{align}
p(z_+)&=-\frac{1}{2}\cos2\theta+(1+\ln2)-i(\theta+\frac{1}{2}\sin2\theta)\\
 p(z_-)&=-\frac{1}{2}\cos2\theta+(1+\ln2)+i(\theta+\frac{1}{2}\sin2\theta-\pi)\\
p_\pm &:=p''(z_\pm)=8e^{\mp i\theta}\cos\theta.\label{eqbulkppm}
\end{align}

We are now ready to evaluate the asymptotic limit of   $\varphi_{\beta,N-l}\big(s; f)$, which is related to that of $I_{N}(\bar s;\bar f)$ via \eqref{HE1}.  We start with the case $n=2m$ and $l=0$.  By substituting Eqs.\ \eqref{eqbulkz}--\eqref{eqbulkppm} and the equality $q(t)=Q(t)$  into Proposition \ref{twosaddle}, we obtain
   \begin{align} I_{N,2m}&\big(\bar s;\bar f \big)\sim (8N)^{-\beta'\! m(m-1)/2-m}\exp\{-nNu^{2} -nN(1+2\ln 2-i\pi)/2\} \nonumber\\
&\times (\sqrt{1-u^{2}})^{\beta'\! m(m+1)/2+(2r-1)m} \tbinom{2m}{m} (\Gamma_{\beta'\!,m})^{2}\prod_{k=1}^{r}(1+\pi_{k}^{2})^{m} \, {\phantom{j}}_{0}\mathcal{F}_0^{(2/\beta'\!)}((-1)^{m},1^{m};i\pi \bar s).\end{align}
Here the notation $(-1)^{m}$ (resp.\ $1^{m}$) means that $-1$ (resp.\ 1) is repeated $m$ times.   As a consequence, we have that
\be \varphi_{\beta,N}\big(s; f)\sim \,
\Psi_{\!^{N,2m}} \gamma_{m}\!( \beta'\!)\prod_{k=1}^{r}(1+\pi_{k}^{2})^{m} \, {\phantom{j}}_{0}\mathcal{F}_0^{(2/\beta'\!)}((-1)^{m},1^{m};i\pi \bar s).
\ee
 The factors  $\Psi_{\!^{N,2m}}$ and $\gamma_{m}\!( \beta'\!)$ are given in  Appendix \ref{appendix}. Application of formula \eqref{vip2}  then establishes equation \eqref{theobulkeq1} of Theorem \ref{theobulk}.

The case $n=2m-1$ and $l=1$ is very similar.   Combining Eqs.\ \eqref{eqbulkz}--\eqref{eqbulkppm},  the equality $q(t)=Q(t)$, and Proposition \ref{twosaddle}, we get
\begin{align} &I_{N,2m-1}\big(\bar s;\bar f\big)\sim (8N)^{-\beta'\! (m-1)^{2}/2-n/2}\exp\{-nNu^{2} -nN(1+2\ln 2-i\pi)/2\}\nonumber\\
&\times (\sqrt{1-u^{2}})^{\beta'\! (m^{2}-1)/2-n/2+rn} \tbinom{2m-1}{m} \Gamma_{\beta'\!,m-1} \Gamma_{\beta'\!,m}\, (-2i)^{(2m-1)l}\nonumber\\
&\times
\left(e^{-i\theta_{N}+il(\theta+ \pi/2)+ir\theta} \prod_{k=1}^{r}(1-i\pi_{k})^{m-1}(1+i\pi_{k})^{m}\, {\phantom{j}}_{0}\mathcal{F}_0^{(2/\beta'\!)}((-1)^{m-1},1^{m};-i\pi\bar s)
 +(i\rightarrow -i)\right),\end{align}
where  $$\theta_{N}=N(2\theta+\sin2\theta-\pi)/2+\theta(1+(m-1)\beta')/2, \qquad \theta=\arcsin u.$$
Hence,
\begin{align} &\varphi_{\beta,N-l}\big(s;f\big)\sim \Psi^{\!_{(l)}}_{\!^{N,2m-1}}  \frac{1}{2i\sqrt{\cos\theta}}\nonumber\\
&\times\left(e^{-i\theta_{N}+il(\theta+ \pi/2)+ir\theta} \prod_{k=1}^{r}(1-i\pi_{k})^{m-1}(1+i\pi_{k})^{m}\, {\phantom{j}}_{0}\mathcal{F}_0^{(2/\beta'\!)}((-1)^{m-1},1^{m};-i\pi \bar s)
 +(i\rightarrow -i)\right) \label{oddbulkHE}
.\end{align}
From this and Eq.\ \eqref{vip2},  one easily derives Eq.\ \eqref{theobulkeq2} of Theorem \ref{theobulk}.  The coefficient $\Psi^{\!_{(l)}}_{\!^{N,2m-1}}$ is given in  Appendix \ref{appendix}.
\end{proof}

\subsection{Edge limit: sub-critical  regime}

\begin{proof}[Proof of  Theorem \ref{theosoftsub}]  We will prove the asymptotic limit of $\varphi_{\beta,N}\big(s;f\big)$ in the sub-critical regime of the soft-edge.  Our method relies mainly on Proposition \ref{onesaddle}.

We start with Eq.\ \eqref{scaleproof} and set  $$ u=1,\qquad\rho=2N^{-1/3},\qquad  \bar f_k=\pi_k/2 .$$  We also suppose that $\pi_k$ belongs to a compact subset of $(-\infty,1)$.

Next, we go to Eqs.\ \eqref{HE1}--\eqref{HE3}. We set $l=0$.
Given that $u=1$, we know that the soft edge is reached. This situation corresponds to the second case of page \pageref{page3cases}, for which $p(z)$ admits one saddle point $z_0=i/2$ of degree $d-1=2$.  Simple calculations then yield
$$p(z_0)=\ln2+(3-i \pi)/2,\qquad p'''(z_0)=-16i.$$
Moreover, the choice of  $\rho=2N^{-1/3}$ allows us to factorize  $Q(t)$ as follows:
$$ Q(t)=q(t)g(N^{1/3}(t-t_0))$$
where
$$ q(t)=\prod_{j=1}^n\prod_{k=1}^r (t_j-i \pi_k/2)\prod_{j=1}^n t_j^{-r},
\qquad  g(N^{1/3}(t-t_0))=\,{\phantom{j}}_{0}\mathcal{F}_0^{(2/\beta'\!)}(2i\bar s;N^{1/3}( t-i/2)).$$

We are now ready to take the asymptotic limit. According to the above equations and Proposition \ref{onesaddle}, we have as $N\to\infty$,
$$ I_N(\bar s;\bar f)\sim \frac{e^{-nN(\ln2+(3-i \pi)/2)}}{N^{(n+n_{\beta'})/3}}\prod_{k=1}^r (1-\pi_{k})^n\int_{\mathbb{R}^{n}} \exp\Big\{\frac{8i}{3}\sum_{j=1}^n w_j^3\Big\}\, {\phantom{j}}_{0}\mathcal{F}_0^{(2/\beta'\!)}(2i\bar s;w)\,|\Delta_{n}(w)|^{\beta'}\, d^{n}w.  $$
Theorem \ref{theosoftsub} follows from this result together with Eqs.\ \eqref{HE1} and \eqref{Airydef}.
\end{proof}

\subsection{Edge limit: critical regime}

\begin{proof}[Proof of  Theorem \ref{theosoft}]

This case is  similar to the previous one.   Once again, in Eqs.\ \eqref{scaleproof}--\eqref{HE3}, we let $$u=1 , \qquad \bar f_k=\pi_k/2 ,\qquad  l=0.$$  The point  $z_0=i/2$ is still a double saddle point and is such that
$$p(z_0)=\ln2+(3-i \pi)/2,\qquad p'''(z_0)=-16i.$$

This time however, we keep $\pi_{l}$  fixed in a compact subset of $(-\infty,1)$ only for all $m+1\leq l \leq r$, while we set
$$ \pi_k =1+\frac{\bar{\pi}_k}{N^{1/3}}, \qquad \text{for all}\quad 1\leq k \leq m.$$
 This allows a slightly different factorization of $Q(t)$ in Eq.\ \eqref{HE3}:
$$ Q(t)=q(t)g(N^{1/3}(t-t_0))$$
where
$$ q(t)=\prod_{j=1}^n\prod_{k=m+1}^r (t_j-i \pi_{k}/2)\prod_{j=1}^n t_j^{-r}$$ and
$$  g(N^{1/3}(t-t_0))=N^{-nm/3}\prod_{j=1}^n\prod_{k=1}^m (N^{1/3}(t_j-i/2)-i \bar \pi_{k}/2 ) \,{\phantom{j}}_{0}\mathcal{F}_0^{(2/\beta'\!)}(2i\bar s;N^{1/3}( t-i/2)).$$

Finally, making use of Proposition \ref{onesaddle}, we  get
 \begin{multline*} I_N(\bar s;\bar f)\sim \frac{e^{-nN(\ln2+(3-i \pi)/2)}}{N^{(n+n_{\beta'}+nm)/3}}\left(\frac{2}{i}\right)^{nm} \prod_{k=m+1}^r (1-\pi_{k})^n\\\times\int_{\mathbb{R}^{n}} \exp\Big\{\frac{8i}{3}\sum_{j=1}^n w_j^3\Big\}\,\prod_{j=1}^n\prod_{k=1}^m (w_j-i \bar{\pi}_{k}/2)  \,{\phantom{j}}_{0}\mathcal{F}_0^{(2/\beta'\!)}(2i\bar{s};w)\,|\Delta_{n}(w)|^{\beta'}\, d^{n}w. \end{multline*}
Obvious simplifications and the comparison with Eq.\  \eqref{Airydef} complete the proof.
\end{proof}

\subsection{Edge limit: supercritical regime}
\begin{proof}[Proof of  Theorem \ref{theosoftsup}]

The third case  of page \pageref{page3cases} gives the supercritical regime.  More explicitly, when the spectral parameter $u>1$, the saddle points of $p$ are
$$z_+=   i(u+\sqrt{ u^{2}-1})/2,\quad  z_-= i(u-\sqrt{u^{2}-1})/2 .$$
  One easily verifies that
$$ p''(z_+)=8\,{\frac {{u}^{2}+u\,\sqrt {{u}^{2}-1}-1}{  ( u+\sqrt {{u}^{2}-1}
  ) ^{2}}}
, \qquad p''(z_-)=8\,{\frac {{u}^{2}-u\,\sqrt {{u}^{2}-1}-1}{  ( u-\sqrt {{u}^{2}
-1}  ) ^{2}}},
$$
so that
$$ 0< p''(z_+) < 4,\qquad -\infty < p''(z_+) < 0. $$
This implies that for the first saddle point, the angles of steepest descent are $0$ and $\pi$, while for the second, they are $\pm \pi/2$. Given that the original path of integration of each variable $t_{j}$ follows the real line, we see that the path of integration cannot be deformed into a path of steepest descent that would go through both saddle points.  Consequently, we consider    $z_0=z_+$ as a single saddle point of degree one.

Before evaluating the integral, let us simplify the notation by introducing new variables:
\be\label{numusigma} \nu =u+\sqrt{ u^{2}-1},\qquad \mu=\nu+\frac{1}{\nu},\qquad \sigma^2=\frac{\nu^2}{\nu^2-1}.
\ee
Thus,
$$ u=\frac{\mu}{2},\qquad z_0=\frac{i\nu}{2},\qquad p(z_0)=\frac{\mu\nu}{2}-\ln\nu+\ln 2+\frac{1-i\pi}{2},\qquad p''(z_0)=\frac{4}{\sigma^2}.$$
In Eq.\ \eqref{scaleproof}, we also set
$$ \rho=\frac{2\sigma}{N^{1/2}},\qquad \bar f_j=\frac{\nu}{2}+\frac{\sigma  \bar\pi_j}{2N^{1/2}},\qquad 1\leq j\leq m, $$
and suppose that the other spectral variables $\bar f_{m+1}=\pi_{m+1}/2,\ldots,\bar f_r=\pi_{r}/2$ belong to a compact subset of $(-\infty,\nu/2)$.

The function $Q(t)$, which appears in the integrand of $I_N(\bar s;\bar f)$, can now be factorized as
$$ Q(t)=q(t)g(N^{1/2}(t-t_0)),$$
where
$$ q(t)=\prod_{j=1}^n\prod_{k=m+1}^r (t_j-i \pi_{k}/2)\prod_{j=1}^n t_j^{-r}$$ and
\begin{multline*} g(N^{1/2}(t-t_0))=N^{-nm/2}\exp\Big\{-\frac{2\nu-\mu}{2\sigma }N^{1/2}p_1(\bar s)\Big\}
\\\times  \prod_{j=1}^n\prod_{k=1}^m (N^{1/2}(t_j-z_0)-i \sigma \bar\pi_{k}/2 )\ {_{0}\mathcal{F}_0^{(2/\beta'\!)}}(2i\bar s/\sigma;N^{1/2}( t-t_0)).
\end{multline*}
The use of Proposition \ref{onesaddle} then leads to
 \begin{multline*} I_N(\bar s;\bar f)\sim \frac{e^{-nN( {\mu\nu}/{2}-\ln\nu+\ln 2+{(1-i\pi)}/{2})}}{N^{(n+n_{\beta'}+nm)/2}}\left(\frac{2}{i}\right)^{nm} \nu^{-rn}\exp\Big\{-\frac{2\nu-\mu}{2\sigma }N^{1/2}p_1(\bar s)\Big\}\prod_{k=m+1}^r (\nu-\pi_{k})^n\\
\times \int_{\mathbb{R}^{n}} \exp\Big\{-\frac{2}{\sigma^2}\sum_{j=1}^n w_j^2\Big\}\, \prod_{j=1}^n\prod_{k=1}^m (w_j-i \sigma \bar\pi_{k}/2 )\,{_{0}\mathcal{F}_0^{(2/\beta'\!)}}(2i\bar s/\sigma;w)\,|\Delta_{n}(w)|^{\beta'}\, d^{n}w.  \end{multline*}

Comparing with the definition of the multivariate Gaussian function \eqref{Gdef}, we get
 \begin{equation*} I_N(\bar s;\bar f)\sim {\frac {(-1)^{nm}i^{nN}{\nu}^{n(N-r)}\Gamma_{\beta',n}}{{2}^{n(N-m)}
{N}^{(n+n_{\beta'}+nm)/2 }{{\rm e}^{Nn
 \left( \mu\,\nu+1 \right)/2 }} }}\left( \frac{\sigma}{2} \right) ^{n+n_{\beta'} +nm} e^{-\frac{2\nu-\mu}{2\sigma }\sqrt{N} p_1(\bar s)}\prod_{k=m+1}^r (\nu-\pi_{k})^n
G^{(2/\beta')}_{n,m}(\bar s;\bar \pi).  \end{equation*}
Finally, the substitution of the latter equation into with Eq.\ \eqref{HE1} leads to
\begin{multline}\label{finalformforsuper} \varphi_{\beta,N}(s;f)\sim
(-1)^{nm} 2^{-nN/2}  \sigma^{n(1+m)+2n(n-1)/\beta}  \nu^{n(N-r)} N^{n(N-m)/2} e^{-nN(1+( \nu-\mu/2)\mu)/2}
   \times \\
 \exp\Big\{-\frac{2\nu-\mu}{2\sigma}N^{1/2} p_1(\bar s)\Big\} \prod_{k=m+1}^r (\nu-\pi_{k})^n \, e^{\frac{1}{4 \sigma^{2}}p_{2}(\bar s)}
G^{(2/\beta')}_{n,m}(\bar s;\bar \pi),
 \end{multline}
which is equivalent to the expected result.
 \end{proof}

\begin{remark}  We stress that on the RHS of \eqref{finalformforsuper}, the factor $e^{-\frac{2\nu-\mu}{2\sigma }\sqrt{N} p_1(\bar s)}$ is not negligible, even when  $N\rightarrow\infty$.
This differs considerably from what we have observed for the limiting correlations in the subcritical and the critical regimes.  Indeed, in these regimes, the asymptotic limit of $\varphi_{\beta,N}$ factorizes as  product of one function depending on $\beta,N,n,m$  and  another function  depending the $\bar{s}_j$'s and $\bar{\pi}_j$'s, but independent of $N$.

One reason of causing the  difference in the  asymptotic behaviors may come from the weighted factor $e^{-\frac{1}{2}p_{2}(s)}$ in
Eq. \eqref{weightedquantity}.  As a matter of fact,  if we replace  the weighted quantity  by
 $$\hat{\varphi}_{\beta,N}(s;f)=e^{ \frac{\nu^{2}-1}{2(\nu^{2}+1)}p_{2}(s)}\varphi_{\beta,N}(s;f) =e^{ -\frac{1}{ \nu^{2}+1}p_{2}(s)} K_{\beta,N}(s;f)$$
 in the supercritical regime, then with the same scalings, we  get
 \begin{multline}\label{hat} \hat{\varphi}_{\beta,N}(s;f)\sim
(-1)^{nm} 2^{-nN/2}  \sigma^{n(1+m)+2n(n-1)/\beta}  \nu^{n(N-r)} N^{n(N-m)/2} e^{-nN/2}
   \times \\
  \prod_{k=m+1}^r (\nu-\pi_{k})^n \, e^{\frac{\nu^{2}-1}{2(\nu^{2}+1)}p_{2}(\bar s)}
G^{(2/\beta')}_{n,m}(\bar s;\bar \pi).
 \end{multline}

Thus, in order to rewrite the three regimes in a consistent way, we could   introduce the following function:
 \be \nu(u)=\begin{cases} 1, & |u|\leq 1,  \\ u+\sqrt{u^{2}-1},  & |u|>1.  \end{cases}  \nonumber\ee
This would allow us to write $$\varphi_{\beta,N}(s;f) =e^{ -\frac{1}{ \nu^{2}+1}p_{2}(s)} K_{\beta,N}(s;f).  $$
The asymptotic behavior of $\varphi_{\beta,N}(s;f)$  in the subcritical and critical regimes would be the same as in Theorems \ref{theosoftsub} and \ref{theosoft}.  However,  in the supercritical regime, Eq.\ \eqref{theosoftsupeq} of Theorem \ref {theosoftsup}   would be replaced by Eq.\ \eqref{hat}.    \end{remark}

\subsection{Slowly growing rank case}\label{sectionslow}

In Remark \ref{growingrank}, we claim that Theorems \ref{theosoftsub} to  \ref{theobulk}  still hold in the case where $r$ grows sufficiently slowly with $N$.  This can be understood as follows.

Suppose first that in Eq.\ \eqref{HE1}, the function $p(z)$ has a saddle point $z_0$ of order $d-1$.  Recall that $d=3$ in the subcritical and critical regimes, while $d=2$ in the supercritical regime and in the bulk.  In the neighborhood of the saddle point, let
\be t_j=z_0+\frac{1}{N^b}w_j,\qquad b=\frac{1}{d}.\label{changegrowth}\ee  Then,
\be \label{growth} N(p(t_j)-p(z_0))=\frac{p^{(d)}(z_0)}{d!}\,  w_j^d +O(N^{-b})\ee
It is worth stressing that Eqs.\ \eqref{changegrowth} and \eqref{growth} are basic steps in the proof of Propositions \ref{onesaddle} and \ref{twosaddle}.  Moreover, under the  change \eqref{changegrowth}, the factor $t_j^{-r}\prod_{k=1}^r (t_j-i \bar{f}_{k})$  coming from the function $Q(t)$ of Eq.~\eqref{HE3}  becomes
\be \label{prodgrowth} \big(z_0+\frac{w_j}{N^b}\big)^{-r}\prod_{k=1}^r \big(z_0 -i \bar{f}_{k}+\frac{w_j}{N^b}\big).\ee

Now, suppose the following asymptotic growth of $r$ as $N\to \infty$:
\be r\sim R\, N^a \ee
for some  non-negative constants $R$ and $a$.   For fixed values of the variables $w_j$ and $\bar{f}_k$, the product \eqref{prodgrowth} remains finite as $N\to\infty$ if $a<b$.  In other words, the equation
\be \label{condition} \lim_{N\to\infty}\frac{r}{N^b}=0\ee  is a sufficient condition that guarantees the non-growing behavior of \eqref{prodgrowth} as $N\to\infty$.   Given that \eqref{prodgrowth} is well defined whenever \eqref{condition} holds, one can apply Propositions \ref{onesaddle} and \ref{twosaddle} as if the rank $r$ was finite.


\subsection{Proof of  Proposition \ref{propAiry}}
\begin{proof}[Proof of  Proposition \ref{propAiry}]
\eqref{airy1} and \eqref{airy2} originate from integrals evaluated around  one  simple saddle point and two simple saddle points, respectively.  For \eqref{airy1},  set $N=x^{3/2}$.  Simple manipulations and the use of   \eqref{vip1} lead to
\begin{align} \mathrm{Ai}_{n,r}^{(\alpha)}(x+x^{-1/2}s;x^{1/2}\bar f)&=(2\pi)^{-n}N^{((1+r)n+n(n-1)/\alpha)/3}\nonumber\\
&\times\int_{\mathbb{R}^n}e^{-N\sum_{j}p(t_{j})} |\Delta_{n}(t)|^{2/\alpha} \prod_{j=1}^n\prod_{l=1}^r (i t_j+\bar f_{l})\,\!{\phantom{j}}_0\mathcal{F}_0^{(\alpha)}(s;it) d^{n}t \end{align}
where $p(t_j)=-it_j^3/3-it_j$.

 The function $p(t_{j})$ has two simple saddle points at $\pm i$.  With $z_0=i$, we have $p''(z_0)=2$ which implies that the steepest descent path near $z_0$ would follow the horizonal  line, as desired.  We thus have an integral like in Proposition \ref{onesaddle} with  $$q(t)=\prod_{j=1}^n\prod_{l=k+1}^r (i t_j+ f_{l}) \!{\phantom{j}}_0\mathcal{F}_0^{(\alpha)}(s;it),\qquad g(N^{1/2}(t-t_0))=\prod_{j=1}^n\prod_{l=1}^k (i\sqrt{2N} (t_j-z_0)+ f_{l}).$$ We may thus apply  Proposition \ref{onesaddle} to the case $\mu=2$, $z_0=i$,   $p(z_0)=2/3$  and \eqref{airy1} follows immediately.

For \eqref{airy2}, we also let $N={x}^{3/2}$.  In the definition of $\mathrm{Ai^{(\alpha)}_{n,r}}(s)$, substitute $t_j$ by  $N^{1/3}t_j$ and  apply \eqref{vip1}, so we get
\begin{align} \mathrm{Ai}_{n,r}^{(\alpha)}(-x+x^{-1/2}s;x^{1/2}  f)&=(2\pi)^{-n}N^{((1+r)n+n(n-1)/\alpha)/3}\nonumber\\
&\times\int_{\mathbb{R}^n}e^{-N\sum_{j}p(t_{j})} |\Delta_{n}(t)|^{2/\alpha} \prod_{j=1}^n\prod_{l=1}^r (i t_j+ f_{l})\,\!{\phantom{j}}_0\mathcal{F}_0^{(\alpha)}(s;it) d^{n}t \end{align}
where $p(t_j)=-it_j^3/3+it_j$.
This function has 2 simple saddle points, namely $x_\pm=\pm 1$. This time we have to consider both of them because they are already on the path of integration.  We have $p(x_\pm)=\pm 2i/3$, $p_\pm=p''(x_\pm)=\mp 2i$.  This means that the steepest descent path is  given by
\be \mathscr{P}=\left\{-1+\tau e^{-i\pi/4}\,:\,\tau \in(-\infty,\sqrt{2}]\right\}\cup \left\{1+\tau e^{i\pi/4}\,:\, \tau\in[-\sqrt{2},\infty)\right\}.\nonumber\ee
Thus  \eqref{airy2} follows from Proposition \ref{twosaddle}.
\end{proof}

\section{Conclusion}
The scaling limits of correlations of  characteristic polynomials for the Gaussian $\beta$-ensemble, perturbed by a finite rank matrix source,   have been computed.   In particular, at the soft edge of the spectrum,   two distinct families of  multivariate  functions have been proved to be the scaling limits in the (sub)critical and supercritical regimes, so a phase transition phenomenon has been observed. To our knowledge,   even in the case of $\beta=1,4$ the  results obtained in this paper are new.

The duality formula \eqref{duality} for the Gaussian ensemble plays a key role in our asymptotic analysis. A similar formula holds for the chiral Gaussian ensemble\cite{des}, so a future challenging problem is to compute the corresponding limit of characteristic polynomials and   show that it is the same as in the Gaussian case (some universal pattern).   As a matter of fact,  Forrester \cite{for2012} has obtained the soft-edge limit of one single characteristic polynomial  not only for the Gaussian case (any finite rank $r$), but also for the chiral case (only rank 1).  Moreover,  it has been proved  in the previous paper \cite{dl} that both ensembles without source indeed share the same scaling limit -- there the duality formula for the  chiral case is not used.  Those  observations suggest  that the same phase transition phenomenon might still  hold for the chiral case.

\begin{acknow}

The work of D.-Z.~Liu\ was  supported by the National Natural Science Foundation of China grant \#11301499 and by CUSF WK 0010000026. The work of P.~D.\ was  supported by CONICYT's  Anillo de Investigaci\'on ACT56 and FONDECYT's grants \#1090034 and \#1131098.  P.D. is also grateful to D.~C\^ot\'e (CRIUSMQ)   for recent financial support.
\end{acknow}

  \begin{appendix}

\section{Notation and  constants}
\label{appendix}
First of all, the normalization constant for the Gaussian $\beta$-ensemble is equal to (see e.g., \cite{forrester})
 \be \label{constantforg} G_{\beta,N}=\beta^{-N/2-\beta N(N-1)/4}(2\pi)^{N/2}\prod_{j=0}^{N-1}
\frac{\Gamma(1+\beta/2+j\beta/2)}
{\Gamma(1+\beta/2)}.\ee
This constant is in fact a special case of the following integral:
\be\label{gaussintgen} \int_{\mathbb{R}^n}\prod_{i=1}^n e^{-zx_i^2/2}\prod_{1\leq i<j \leq n}|x_i-x_j|^{\beta}\,dx_1\cdots dx_n=\frac{1}{z^{(n+\beta n(n-1)/2)/2}}\Gamma_{\beta,n},\qquad \textrm{Re}\{z\}>0,
\ee
where
\be \label{constgamma}
 \Gamma_{\beta,n}=(2\pi)^{n/2}\prod_{j=1}^n\frac{\Gamma(1+j\beta/2)}{\Gamma(1+\beta/2)}.
\ee
The duality formula \eqref{duality} involves the latter factor as follows:
\be \label{dualityconstant} D_{\!\beta,N,n}=\frac{i^{nN}2^{n/2+ n(n-1)/\beta}}{ \Gamma_{4/\beta\!,n}} .\ee

We now consider the constants related to the asymptotic behavior of the average product of characteristic polynomials at the soft edge. In the sub-critical regime, we have
\be\label{eqPhisub}
\Phi_{\!^{\beta,N,n}}=\frac{{\pi}^{n}
{N}^{\,{
\frac {n\, \left( 3\,N\,\beta+\beta+2\,n-2 \right) }{6\beta}}}
}{  \Gamma_{4/\beta,n}\, {{ e}^{nN/2}}\, {2}^{\,n\, \left( N-2 \right)/2 }}.
\ee
For the critical regime, the constant is
\be\label{eqPhi}
\Phi_{\!^{\beta,N,n,m}}=  (-1) ^{n\,m} {N}^{-
\,n\,m/3} \Phi_{\!^{\beta,N,n}}.
\ee
For the supercritical regime,  new positive parameters are needed: $\nu$, $\mu=\nu+\nu^{-1}$, and $\sigma^2=\nu^2/(\nu^2-1)$.  The constant then reads:
\begin{equation} \label{eqPhisup} \Phi_{\!^{\beta,N,n,m}}^\text{sup}= \left( -1 \right) ^{nm }{{ e}^{ \frac{\,nN \left( {\mu}^{2}-2\,\mu\,\nu-2 \right)}{4}
}   {\sigma}^{{\frac {n \left( 2\,n+\beta -2+m\beta  \right) }{\beta }}}{N
}^{\frac{n \left( N-m \right)}{2} }   }{2}^{-\frac{nN}{2}}  {\nu}  ^{n(N-r)}. 
\end{equation}


In the bulk of the spectrum, the constant is, for $n=2m$,
\be\label{eqPsi}
\Psi_{\!^{N,2m}}= 2^{\beta'\! m(m+1)/2-m(N+1)} N^{\beta'\! m^{2}/2+mN}e^{ -m N}(\sqrt{1-u^{2}})^{\beta'\! m(m+1)/2-m+2mr} ,
\ee while for $n=2m-1$,
\begin{multline}\label{eqPsiodd}
\Psi_{\!^{N,2m-1}}^{_{(l)}}= \tbinom{2m-1}{m}\frac{ \Gamma_{\beta'\!,m-1} \Gamma_{\beta'\!,m}}{\Gamma_{\beta'\!,2m-1}}
2^{\beta'\!(m^{2}-1)/2-(2m-1)(N+1-l)/2}  N^{\beta'\! m(m-1)/2+(2m-1)(N-l)/2} \\ \times e^{ -(2m-1) N/2}(\sqrt{1-u^{2}})^{\beta'\!(m^{2}-1)/2-(2m-1)/2+nr}  \,(2i \sqrt[4]{1-u^{2}} ).
\end{multline}

Finally, the universal coefficient is  \be \gamma_{m}(\beta'\!)=\binom{2m}{m}\prod_{j=1}^m\frac{\Gamma(1+\beta'\!j/2)}{\Gamma(1+\beta'\!(m+j)/2)}.\label{univcoefficientevenbulk}\ee
  \end{appendix}

\end{document}